\documentclass[10pt]{article}

\usepackage{amsmath,amssymb,amsthm}
\usepackage{graphicx}
\usepackage{subfigure}
\usepackage{enumerate}
\usepackage{url}
\usepackage{float}
\usepackage{xspace}
\usepackage{hyperref}
\usepackage{color}
\usepackage{thmtools, thm-restate}

\newcommand{\R}{\mathbb{R}}

\newcommand{\N}{\mathbb{N}}
\newtheorem{theorem}{Theorem}

\newtheorem{lemma}{Lemma}[section]

\newtheorem{definition}{Definition}

\begin{document}

\author{K\'aroly J. B\"or\"oczky\footnote{Alfr\'ed R\'enyi Institute of Mathematics, Hungarian Academy
  of Sciences, Re\'altanoda u. 13-15., H-1053 Budapest, Hungary, boroczky.karoly.j@renyi.hu, Research is supported in parts by NKFIH grant 132002.},
Alexey Glazyrin \footnote{The University of Texas Rio Grande Valley, School of Mathematical \& Statistical Sciences, One West University Blvd, Brownsville, Texas, USA, Alexey.Glazyrin@utrgv.edu, Research is partially supported by NSF grant DMS-2054536}}

\date{}

\title{Stability of optimal spherical codes}

\maketitle

\begin{abstract}
For many extremal configurations of points on a sphere, the linear programming approach can be used to show their optimality. In this paper we establish the general framework for showing stability of such configurations and use this framework to prove the stability of the two spherical codes formed by minimal vectors of the lattice $E_8$ and of the Leech lattice.
\end{abstract}

\section{Definitions and main results}

By a spherical $d$-dimensional code we mean a finite set of points from the unit sphere $\mathbb{S}^{d-1}$. A $d$-dimensional spherical code with $N$ points is called a $(d,N,s)$-code if all pairwise dot products of distinct points from the code are not greater than $s$. When we say that a $(d,N,s)$-code is optimal, we mean that there doesn't exist a $(d,N',s)$-code with $N'>N$.

In this paper, we consider optimal spherical codes whose optimality can be shown via the linear programming bound. In particular, we want to concentrate on two classical codes, the $(8,240,1/2)$-code and the $(24,196560,1/2)$-code. The optimality of these codes was shown independently by Odlyzko and Sloane \cite{odl79} and by Levenshtein \cite{lev79}. Bannai and Sloane \cite{BaS81} proved that
both the $(8,240,1/2)$-code and the $(24,196560,1/2)$-code are unique up to orthogonal transformations in their respective spaces.  Actually, assuming that the non-zero vectors of minimal length of the corresponding lattices are of unit length, the $(8,240,1/2)$-code consists of $E_8\cap S^7$, and the $(24,196560,1/2)$-code consists of
$\Lambda_{24}\cap S^{23}$, where $\Lambda_{24}$ is the Leech lattice. These codes solve the kissing number problem in $\R^8$ and $\R^{24}$ (see Conway, Sloane \cite{CoS98} and Erikson, Zinoviev \cite{ErZ01}); namely, the maximum number of non-overlapping unit balls touching a given unit ball is $240$ in $\R^8$ and
196560 in $\R^{24}$. Some recent related results include
 Balla,  Draxler, Keevash, Sudakov \cite{BDKS18} about spherical codes, and 
Keevash, Long \cite{KeL20} about stability of codes.

We will call two symmetric matrices $P, Q$ of the same size $\delta$-close if $||P-Q||_{max}\leq\delta$.

\begin{definition}
Two spherical $d$-dimensional codes $A=\{a_1,\ldots,a_k\}$ and $B=\{b_1,\ldots,b_k\}$ are called $\delta$-close if there is a permutation $\sigma$ on $B$ such that the Gram matrices of $A$ and $\sigma(B)$ are $\delta$-close.
\end{definition}

As in the papers and books referenced above, our approach is based on the linear programming bound. For the linear programming bound for sphere packings on $\mathbb{S}^{d-1}$ we define Gegenbauer polynomials $Q_i$, $i\in\N$, in one variable where each $Q_i$ is of degree $i$, and satisfies the following recursion:
\begin{eqnarray*}
Q_0(t)&=&1\\
Q_1(t)&=&t\\
Q_{i+1}(t)&=&\frac{(2i+d-2)tQ_i(t)-iQ_{i-1}(t)}{i+d-2} \mbox{ \ for $i\geq 2$}.
\end{eqnarray*}
We do not signal the dependence of $Q_i$ on $d$ because the original notation for the Gegenbaur polynomial is
$Q_i=Q_i^{(\alpha)}$ for $\alpha=\frac{d-2}2$ as
$$
\int_{-1}^1Q_i(t)Q_j(t)(1-t^2)^{\frac{d-3}2}\,dt=0\mbox{ \ \ if $i\neq j$}.
$$

Polynomials are normalized so that $Q_i(1)=1$ for all $i$. The main property of these polynomials is that for any spherical code $\{a_1,\ldots,a_k\}\subset \mathbb{S}^{d-1}$ and any non-negative $i$, the $k\times k$ matrix $Q_i(\langle a_m,a_n\rangle)$ is positive semi-definite (see Schoenberg \cite{sch42} or the book of Erikson and Zinoviev \cite{ErZ01}).

We use the following version of the linear programming bound.

\begin{theorem}
\label{linprogbound}
Let $d\geq 2$. If $f=f_0Q_0+f_1Q_1+\ldots+f_kQ_k$ for $k\geq 1$ and
 $f_1,\ldots,f_k\geq 0$, then, for a spherical code $X$ with $N$ points,
\begin{equation}
\label{linprogbound0}
Nf(1)+\sum_{x,y\in X\atop x\neq y}f(\langle x,y\rangle)\geq N^2 f_0.
\end{equation}
\end{theorem}

\begin{proof}
The $N\times N$ matrix formed by $(f-f_0Q_0) (\langle x,y\rangle)$ for all $x,y\in X$ must be positive semi-definite. Hence its sum of elements is non-negative. Since

$$\sum_{x,y\in X}f(\langle x,y\rangle) = N f(1) + \sum_{x,y\in X\atop x\neq y}f(\langle x,y\rangle)$$

and

$$\sum_{x,y\in X}f_0Q_0(\langle x,y\rangle)=N^2f_0,$$

the statement of the theorem follows immediately.
\end{proof}

The classical linear programming bound (sometimes called the Delsarte bound) for $(d,N,s)$-codes is a simple corollary of Theorem \ref{linprogbound}: if we additionally require $f_0>0$ and $f(t)\leq 0$ for all $t\in[-1,s]$, then for any $(d,N,s)$-code
\begin{equation}
\label{linprogbound1}
N\leq f(1)/ f_0,
\end{equation}

because all $f(\langle x,y\rangle)$ will be non-positive for $x\neq y$.

\begin{definition}
\label{Delsarte-tight}
A spherical $(d,N,s)$-code is called Delsarte-tight if there exists a polynomial $f=f_0Q_0+f_1Q_1+\ldots+f_kQ_k$ for $k\geq 1$, $f_0>0$ and $f_1,\ldots,f_k\geq 0$, such that $f(t)\leq 0$ for all $t\in[-1,s]$ and $N=f(1)/ f_0$.
\end{definition}

\begin{table}
\caption{Table of the known sharp configurations, together with the $600$-cell (from \cite{coh07}).}\label{tab:sharp}
\begin{center}
\begin{tabular}{ccccc}
$n$ & $N$ & $M$ & Inner products & Name\\
\hline $2$ & $N$ & $N-1$ & $\cos (2\pi j/N)$ ($1 \le j \le
N/2$) & $N$-gon\\
$n$ & $N \le n$ & $1$ & $-1/(N-1)$ & simplex\\
$n$ & $n+1$ & $2$ & $-1/n$ & simplex\\
$n$ & $2n$ & $3$ & $-1, 0$ & cross polytope\\
$3$ & $12$ & $5$ & $-1, \pm 1/\sqrt{5}$ & icosahedron\\
$4$ & $120$ & $11$ & $-1,\pm 1/2, 0,(\pm 1 \pm \sqrt{5})/4$ & $600$-cell\\
$8$ & $240$ & $7$ & $-1,\pm 1/2, 0$ & $E_8$ roots\\
$7$ & $56$ & $5$ & $-1, \pm 1/3$ & kissing\\
$6$ & $27$ & $4$ & $-1/2, 1/4$ & kissing/Schl\"afli\\
$5$ & $16$ & $3$ & $-3/5, 1/5$ & kissing\\
$24$ & $196560$ & $11$ & $-1,\pm 1/2, \pm 1/4, 0$ & Leech lattice\\
$23$ & $4600$ & $7$ & $-1,\pm 1/3, 0$ & kissing\\
$22$ & $891$ & $5$ & $-1/2, -1/8, 1/4$ & kissing\\
$23$ & $552$ & $5$ & $-1, \pm 1/5$ & equiangular lines\\
$22$ & $275$ & $4$ & $-1/4, 1/6$ & kissing\\
$21$ & $162$ & $3$ & $-2/7, 1/7$ & kissing\\
$22$ & $100$ & $3$ & $-4/11, 1/11$ & Higman-Sims\\
$q\frac{q^3+1}{q+1}$ & $(q+1)(q^3+1)$ & $3$ & $-1/q, 1/q^2$ &
{\hskip -8pt}
isotropic subspaces\\
& & {\hskip -22pt}(4 if $q=2$){\hskip -22pt}
& & {\hskip -8pt} ($q$ a prime power)\\
\\
\end{tabular}
\end{center}
\end{table}

In the case all coefficients $f_1,\ldots,f_k$ are strictly positive, a Delsarte-tight set is known in literature as a sharp set (see \cite{ban09} for more details). This is not always the case that all these coefficients are strictly positive. For instance, the vertices of a 600-cell in $\mathbb{S}^3$ form a Delsarte-tight but not sharp set. Table \ref{tab:sharp} from \cite{coh07} lists all known sharp configurations with their inner products and types.

The main goal of this paper is to find out what codes satisfy the relaxed condition on the maximal inner product if the strict condition defines a Delsarte-tight $(d,N,s)$-code. We consider $(d,N,s+\varepsilon)$-codes and show how close they must be to Delsarte-tight $(d,N,s)$-codes.

\begin{theorem}[Weak stability of Delsarte-tight codes]\label{thm:general}
If there exists a Delsarte-tight $(d,N,s)$-code, then, for sufficiently small positive $\varepsilon$, any $d$-dimensional spherical code with all pairwise dot products less than $s+\varepsilon$ has no more than $N$ points and for any $(d,N,s+\varepsilon)$-code $S$ there is a constant $C$ and an isometry $A$ of $\mathbb{S}^{d-1}$ such that pairwise spherical distances between the points of $A(S)$ and the points of some Delsarte-tight $(d,N,s)$-code $T$ are not greater than $C\varepsilon^{1/m}$, where $m$ is the largest root multiplicity of the polynomial $f$ corresponding to a Delsarte-tight code $T$ as described in Definition \ref{Delsarte-tight}.
\end{theorem}

In all known examples of Delsarte-tight spherical codes (see Table \ref{tab:sharp}), the largest root multiplicity $m$ of their corresponding polynomials is 1 or 2. Theorem \ref{thm:general} is applicable to all codes from Table \ref{tab:sharp}. Essentially for these codes, the theorem means 1/2-H\"older contunuity of optimal codes depending on the maximal inner product.

For several optimal spherical codes (simplex, cross polytope, icosahedron, 600-cell), their optimality can be show using the so-called simplex bound. The strong stability of the simplex bound for these codes was shown in \cite{BBGK}.

As the main results of this paper we show the strong stability of the linear programming bound for the $(8,240,1/2)$-code and the $(24,196560,1/2)$-code.

\begin{theorem}[Strong stability of the $(8,240,1/2)$-code]\label{thm:8}
For sufficiently small positive $\varepsilon$, any $8$-dimensional spherical code with all pairwise dot products less than $1/2+\varepsilon$ has no more than $240$ points and for any $(8,240,1/2+\varepsilon)$-code $S$ there is a constant $C_8$ and an isometry $A$ of $\mathbb{S}^{7}$ such that pairwise spherical distances between the points of $A(S)$ and the points of the $(8,240,1/2)$-code are not greater than $C_8\varepsilon$.
\end{theorem}

\begin{theorem}[Strong stability of the $(24,196560,1/2)$-code]\label{thm:24}
For sufficiently small positive $\varepsilon$, any $24$-dimensional spherical code with all pairwise dot products less than $1/2+\varepsilon$ has no more than $196560$ points and for any $(24,196560,1/2+\varepsilon)$-code $S$ there is a constant $C_{24}$ and an isometry $A$ of $\mathbb{S}^{23}$ such that pairwise spherical distances between the points of $A(S)$ and the points of the $(24,196560,1/2)$-code are not greater than $C_{24}\varepsilon$.
\end{theorem}

We also find concrete Lipschitz constants $C_8$ and $C_{24}$.

The paper is organized as follows. In Section \ref{sect:Gram} we prove the weak stability of Gram matrices of Delsarte-tight codes. Section \ref{sect:eigendecomp} is devoted to the stability of eigendecompositions of positive semi-definite matrices. These two sections combined give the proof of Theorem \ref{thm:general}. In Sections \ref{sect:8} and \ref{sect:24} we analyze the $(8,240,1/2)$- and $(24,196560,1/2)$-codes and obtain the strong stability of these codes subsequently proving Theorems \ref{thm:8} and \ref{thm:24}, respectively.

%
%
%
%
%

\section{Weak stability of Gram matrices for Delsarte-tight codes}\label{sect:Gram}

\begin{lemma}\label{lem:d-code}
A symmetric real matrix with 1's on the diagonal is a Gram matrix of a $d$-dimensional spherical code if and only if all principal minors of this matrix of size greater than $d$ are 0 and all principal minors of this matrix of size not greater than $d$ are non-negative.
\end{lemma}

\begin{proof}
Since all principal minors are non-negative, by Sylvester's criterion, the matrix is positive semi-definite and thus is a Gram matrix of a spherical code. If a symmetric matrix has rank $r$, then at least one principal minor of size $r$ is not 0 \cite[Chapter VI, Theorem 4]{wed34}. Any principal minor of size greater than $d$ is 0. Therefore, the rank of the matrix is not greater than $d$ and the spherical code is $d$-dimensional.
\end{proof}

We consider a fixed set $X=\{x_1,\ldots,x_l\}$ of real numbers and a set $M_X$ of all square matrices of size not greater than $N$ with all their entries from $X$.

\begin{lemma}\label{lem:close}
There exists $\varepsilon=\varepsilon(X,N)$ such that for any $\delta<\varepsilon$ and any square matrix $A$ of size not greater than $N$ which is $\delta$-close to a matrix $B$ from $M_X$, the following conditions hold: 1) if $\det A = 0$, then $\det B = 0$; 2) if $\det A \geq 0$, then $\det B\geq 0$.
\end{lemma}

\begin{proof}
There are finitely many matrices in $M_X$ so there is the minimal non-zero value $K$ among all $|\det B|$, $B\in M_X$. Due to the Lipschitz continuity of determinants, there exists also a constant $C$ such that $|\det A - \det B|\leq C\delta$ for any matrix $B$, $B\in M_X$, and any $A$ $\delta$-close to $B$. It is clear that, if we choose $\varepsilon=K/C$, both conditions 1) and 2) must hold because the determinants of $A$ and $B$ will differ by less than $K$.
\end{proof}

\begin{theorem}\label{thm:weak-stability}
If there exists a Delsarte-tight $(d,N,s)$-code, then, for sufficiently small positive $\varepsilon$, any $d$-dimensional spherical code with all pairwise dot products less than $s+\varepsilon$ has no more than $N$ points and for any $(d,N,s+\varepsilon)$-code $S$ there is a constant $K$ such that this code is $K\varepsilon^{1/m}$-close to some $(d,N,s)$-code, where $m$ is the largest root multiplicity of the polynomial $f$ described in Definition \ref{Delsarte-tight}.
\end{theorem}

\begin{proof}

We will begin the proof with several properties of Delsarte-tight sets.

If we have equality in (\ref{linprogbound1}), then (\ref{linprogbound0}) shows that all values $\langle x,y\rangle$ for $x\neq y$, $x,y\in X$ are roots of $f$. Denote the set of roots of $f$ from the segment $[-1,s]$ by $x_1, \ldots, x_l$ and define $R=\{1, x_1,\ldots,x_l\}$. By $M_R$ we mean the set of all matrices of size no greater than $N$ with all entries from $R$.

From the proof of Theorem \ref{linprogbound} we can also conclude that the sum of elements of each matrix $Q_i(\langle x,y\rangle)$ is 0 in case the coefficient $f_i$ of $Q_i$ in the Gegenbauer expansion of $f$ is strictly positive.

In case $f(s)<0$, we can choose any $\varepsilon$ such that $f$ is negative on $[s,s+\varepsilon]$. The linear programming bound will work for $s+\varepsilon$ as well and, since there are no roots of $f$ on $[s,s+\varepsilon]$, any $(d,N,s+\varepsilon)$-code is a $(d,N,s)$-code too. Hence we can assume that $f(s)=0$.

Denote the minimal root of $f$ from $(s,1)$ by $r$ if such a root exists. Otherwise, we take $r=1$. We will denote by $M$ the maximal value of $f(t)/(t-s)$ on $[s,\frac {s+r} 2]$. Then for any positive $\varepsilon$, $\varepsilon\leq \frac {r-s} 2$, the value of $f$ on $[s,s+\varepsilon]$ is not greater than $M\varepsilon$.

Assume we had a $(d,N',s+\varepsilon)$-code $S$ with $N'\geq N$. We want to prove the first part of the theorem and show that $N'=N$ if $\varepsilon$ is small enough. For any pair $x,y$ of distinct points from $S$, the value of $f(\langle x,y\rangle)$ is either non-positive if $\langle x,y\rangle\in[-1,s]$ or, as we have shown above, not greater than $M\varepsilon$ if $\langle x,y\rangle\in[s,s+\varepsilon]$. This implies that

$$\sum_{x,y\in X\atop x\neq y}f(\langle x,y\rangle)\leq N'(N'-1)M\varepsilon.$$

Combining this inequality with Theorem \ref{linprogbound} we get

$$N'f(1)+N'(N'-1)M\varepsilon\geq N'^2 f_0,$$

$$f(1)+(N'-1)M\varepsilon\geq N' f_0.$$

From the tightness of the linear programming bound for the $(d,N,s)$-code we know that $f(1)=Nf_0$. If we assume $N'>N$, we get

$$\varepsilon \geq \frac {N'-N}{N'-1}\cdot \frac {f_0}{M} \geq \frac 1 N\cdot \frac {f_0}{M},$$

which doesn't hold for $\varepsilon< \frac {f_0}{N M}$.

From this moment on, we consider only $\varepsilon< \min\left\{\frac {f_0}{N M},\frac {r-s} 2\right\}$ to prove the second part of the theorem. From Theorem \ref{linprogbound} we conclude that $\sum\limits_{x,y\in X\atop x\neq y}f(\langle x,y\rangle)=0$. All positive elements of this sum are not greater than $M\varepsilon$ so any negative element is at least $-(N^2-N-1)M\varepsilon$.

Denote the multiplicities of roots $x_1,\ldots,x_l$ of $f$ by $m_1,\ldots,m_l$, respectively. We define $M'$ as the minimum of $\max\left\{\left|\frac{f(t)}{(t-x_1)^{m_1}}\right|,\ldots,\left|\frac{f(t)}{(t-x_l)^{m_l}}\right|\right\}$ over all $t\in[-1,1]$. By definition $M'> 0$. For each point $t\in[-1,s]$ there is $x_j$ such that $f(t)\leq -M' |t-x_j|^{m_j}$. If $t=\langle x,y\rangle$ for distinct points $x,y$ from the $(d,N,s+\varepsilon)$-code then we can combine it with the previous inequality:

$$-M' |t-x_j|^{m_j}\geq f(t)\geq -(N^2-N-1)M\varepsilon,$$

$$|t-x_j|\leq \left((N^2-N-1)\frac M {M'} \varepsilon\right)^{1/m_j} \leq K \varepsilon^{1/m},$$

where $K=\max \left\{(N^2-N-1)\frac M {M'},1\right\}$.

We have just proved that the Gram matrix of the $(d,N,s+\varepsilon)$ code is $K \varepsilon^{1/m}$-close to some matrix $B$ from $M_R$. From Lemma \ref{lem:close}, there exists $\varepsilon_0=\varepsilon_0(R,N)$ such that, whenever $K\varepsilon^{1/m} < \varepsilon_0$, if a minor of the Gram matrix of $X$ is 0, then the corresponding minor of $B$ is 0 and, if a minor of the Gram matrix of $X$ is non-negative, then the corresponding minor of $B$ is non-negative. This means that, for $\varepsilon < \left(\frac{\varepsilon_0}{K}\right)^m$, all principal minors of $B$ of size greater than $d$ are 0 and all principal minors of $B$ of size not greater than $d$ are non-negative. By Lemma \ref{lem:d-code}, $B$ is the Gram matrix of a $d$-dimensional spherical code. All non-diagonal entries of $B$ are roots of $f$ so the code is a Delsarte-tight $(d,N,s)$-code.
\end{proof}

\section{Stability of eigenvectors of positive semi-definite matrices}\label{sect:eigendecomp}

For a $N\times N$ matrix $T$, we 
write $\|T\|$ to denote  its spectral norm
and, for a $v\in\R^d$, we write $\|v\|$ to denote its $l_2$-norm. We say that two $N\times N$ matrices $A=[a_{ij}]$ and $B=[b_{ij}]$ are $\delta$-close for $\delta>0$ if
$|a_{ij}-b_{ij}|\leq \delta$ for $i,j=1,\ldots,d$.  It is well-known  that
$A$ and $B$ are $\|A-B\|$-close on the one hand, and if $A$ and $B$ are $\delta$-close, then
$\|A-B\|\leq N\delta$ on the other hand.

For a linear subspace $L$ in some Euclidean space, we write $\cdot|L$ to denote the orthogonal projection into $L$. The following statement is Lemma~2.1 in B\"or\"oczky, B\"or\"oczky, Glazyrin, Kov\'acs \cite{BBGK}, which is needed in the proof of Theorem~\ref{thm:code-stability}.

\begin{lemma}
\label{crosspolytope}
If $d\geq 2$, $\varepsilon\in (0,\frac1{2d})$, and $u_1,\ldots,u_d\in S^{d-1}$ satisfy
$|\langle u_i,u_j\rangle|\leq \varepsilon$ for $i\neq j$, then there exist an orthonormal basis
$w_1,\ldots,w_d$ of $\R^d$ such that $\|u_i-w_i\|\leq 2d\varepsilon$ for $i=1,\ldots,d$.
\end{lemma}

\begin{theorem}\label{thm:code-stability}
For any non-trivial positive semi-definite symmetric matrix $B$ of size $N\geq 2$,  one finds $\delta_0>0$ and $K>0$ depending on $B$ with the following property. If $A$ is a symmetric matrix  of the same rank as $B$ and $\delta$-close to $B$, $0<\delta<\delta_0$, then
 there exist $N\times N$ positive semi-definite symmetric matrices $P$ and $Q$  such that $A=P P$ and $B= Q Q$ where $P$ and $Q$ are $K\delta$-close to each other.
\end{theorem}
\noindent{\bf Remark } We may choose $K=\frac{85N^5\max\{\sqrt{\|B\|},1\}}{\Delta}$ where $\Delta>0$ is the minimum of the minimal gap between consecutive eigenvalues of $B$ and the smallest positive eigenvalue of $B$.

\begin{proof}
  For a positive semi-definite diagonal matrix $T$, we write $\sqrt{T}$ to denote the 
positive semi-definite diagonal matrix whose square is $T$.

Let $r\geq 1$ be the common rank of $A$ and $B$. In addition, let
 $0<\lambda_1<\ldots<\lambda_k$, $k\leq r$, be the different positive eigenvalues of $B$ with corresponding
mutually orthogonal  eigenspaces
$L_1,\ldots,L_k$, $k\leq r$, whose dimensions naturally add up to $r$, and hence 
$L=L_1+\ldots+L_k$ is $r$ dimensional. 
Let $\Delta\in(0,1]$ be maximal such that $\lambda_1\geq \Delta$, and 
$\lambda_j-\lambda_i\geq\Delta$ for $j>i$. We assume that $\delta$ is small enough to have
\begin{equation}
\label{deltacond}
\frac{14N^4}{\Delta}\cdot \delta<\frac1{2N}.
\end{equation}
Let $D$ be the  diagonal matrix $D$ such that the first $r$ diagonal entries are positive and increasing with the index,
and $\lambda_i$ occurs as diagonal entry $\dim L_i$ times.

There exists an $N\times N$ orthogonal matrix $M$ such that $M^{-1}BM=D$, and let $E=M^{-1}AM-D$.  
Since $\|A-B\| \leq N\delta$, we deduce that $\|E\|\leq N\delta$.
We write coordinates in $\R^N$ with respect to the new orthonormal basis obtained via $M$. In particular, $D$ acts on $L_i$ by multiplication by $\lambda_i$, $i=1,\ldots,k$.
For $j=1,\ldots,k$, let $I_j\subset\{1,\ldots,k\}$ be the set of indices of basis vectors contained in $L_j$, and hence
$I_1\cup\ldots\cup I_k=\{1,\ldots,r\}$.

Let $v_1,\ldots,v_N$ be an orthonomal set of eigenvectors of $M^{-1}AM$.
We claim that there exist at least $r$ indices $i\in\{1,\ldots,N\}$ such that $\|v_i|L\|\geq \frac1{\sqrt{N}}$, or in other words, we may reindex $v_1,\ldots,v_N$ in a way such that
\begin{equation}
\label{vprojL}
\|v_i|L\|\geq \frac1{\sqrt{N}}\mbox{ \ if $i\leq r$}.
\end{equation}
We suppose that (\ref{vprojL}) does not hold, and seek a contradicton. Obviously $r<N$ in this case.
The indirect hypothesis yields that there exists a subset $J\subset \{1,\ldots,N\}$ of indices of cardinality
$N-r+1$ 
such that $\|v_i|L\|< \frac1{\sqrt{N}}$ for $i\in J$, and hence
\begin{equation}
\label{vivi}
 \mbox{$\|v_i|L^\bot\|> \sqrt{1-\frac1N}$ for $i\in J$. }
\end{equation}
In addition, if $i\neq j$ for $i,j\in J$, then
\begin{equation}
\label{vivj}
\left|\left\langle (v_i|L^\bot),(v_j|L^\bot)\right\rangle\right|=|\left\langle (v_i|L),(v_j|L)\right\rangle|<\frac1N.
\end{equation}
As $L^\bot$ is $N-r$ dimensional, there exists coefficient $\gamma_i\in\R$ for $i\in I$ with
$\gamma=\max_{i\in J}|\gamma_i|>0$ such that
$\sum_{i\in J}\gamma_i(v_i|L^\bot)=o$. There exists a $j\in J$ such that $\gamma=|\gamma_j|$. We deduce from the triangle inequality, (\ref{vivi}) and (\ref{vivj}) that
\begin{equation}
\label{manyvivj}
0=\left|\left\langle (v_j|L^\bot),\sum_{i\in J}\gamma_i(v_j|L^\bot)\right\rangle\right|>
\gamma\left( 1-\frac1N\right)-(N-r)\gamma\cdot\frac1N\geq 0.
\end{equation}
This contradiction proves (\ref{vprojL}).

For $i\leq r$, we have
\begin{equation}
\label{diagdiff}
\|(\mu_i {\rm Id}-D)v_i\|=\|Ev_i\|\leq N\delta.
\end{equation}
According to (\ref{vprojL}), there exists $\alpha\in I$ such that the $\alpha$th coordinate of $v$ is at least  $1/\sqrt{Nr}$ in absolute value. If $\alpha\in I_j$, then we have
\begin{equation}
\label{lambdamu}
|\mu_i-\lambda_j|\leq \sqrt{Nr}N\cdot\delta\leq N^2\delta,
\end{equation}
and hence (\ref{deltacond}) implies
\begin{equation}
\label{sqrtlambdamu}
|\sqrt{\lambda_j}-\sqrt{\mu_i}|=\frac{|\mu_i-\lambda_j|}{\sqrt{\lambda_j}+\sqrt{\mu_i}}\leq \frac{N^2}{\sqrt{\Delta}}\cdot \delta\leq \frac{N^2}{\Delta}\cdot \delta.
\end{equation}
In addition, (\ref{deltacond}), (\ref{lambdamu}) and the rank of $B$ being $r$ yield
\begin{equation}
\label{muipos}
\mbox{$\mu_i>\Delta/2$ for $i=1,\ldots,r$ and $\mu_i=0$ for $i>r$.}
\end{equation}
If $\lambda_q\neq \lambda_j$, then $|\mu_i-\lambda_q|> \Delta/2$ by (\ref{deltacond}) and (\ref{lambdamu}).
It follows from this observation, from (\ref{lambdamu}) and (\ref{muipos}) that
writing $v_i=(t_1,\ldots,t_N)$, we have
$$
|t_\alpha|\leq  \frac{2N^2}{\Delta}\cdot \delta\mbox{ \ provided $\alpha\not\in I_j$,}
$$
and hence
\begin{equation}
\label{eigenother}
\|v_i|L_j^\bot\|\leq  \frac{2N^2\sqrt{N}}{\Delta}\cdot \delta.
\end{equation}
We conclude from (\ref{deltacond}) and (\ref{eigenother}) that
\begin{equation}
\label{vprojLj}
\|v_i|L_j\|\geq  1-\frac{4N^5}{\Delta^2}\cdot \delta^2>\frac34\\
\end{equation}

We deduce from (\ref{vprojLj}) that for any $i\in I$, there exists a unique
$j(i)\in\{1,\ldots,k\}$ such that $\|v_i|L_{j(i)}\|>\frac34$, and we define $\tilde{v}_i=v_i|L_{j(i)}$.
In particular, if $j=j(i)=j(l)$ for $i,j\in I$, then  (\ref{deltacond}) and (\ref{eigenother}) imply that
\begin{equation}
\label{tildevi}
\begin{array}{rclll}
\|\tilde{v}_i\|&\geq & 1-\frac{4N^5}{\Delta^2}\cdot \delta^2& >1-\frac{\delta}{\Delta}&\\
|\langle \tilde{v}_i,\tilde{v}_l\rangle|&\leq &\frac{4N^5}{\Delta^2}\cdot \delta^2<
\frac{\delta}{\Delta}&&\mbox{ \ if $i\neq l$}.
\end{array}
\end{equation}
Using (\ref{deltacond}) and similar argument as in (\ref{manyvivj}) shows that for any $L_j$, $j=1,\ldots,k$, 
the vectors of the form $\tilde{v}_i$ with $i\in I$ that are contained in $L_j$ are independent, therefore
their number of is at most ${\rm dim}\,L_j$. We deduce from pigeon hole principe that possibly after renumbering
$v_1,\ldots,v_N$, we may asume that $\tilde{v}_i\in L_j$ if and only if $i\in I_j$.

We claim that there exist an orthonormal basis $w_1,\ldots,w_r$ of $L$ such that
if $i\in I_j$, then $w_i\in L_j$ and 
\begin{equation}
\label{wiviL}
\|w_i-v_i\|\leq  \frac{7N^3}{\Delta}\cdot \delta.
\end{equation}
For any $i\leq r$, we set $v'_i=\tilde{v}_i/\|\tilde{v}_i\|\in L_j\cap S^{N-1}$, and hence (\ref{deltacond}) and (\ref{tildevi}) yield that if $i,l\in I_j$, $i\neq l$, then
\begin{equation}
\label{primevi}
|\langle v'_i,v'_l\rangle|\leq \left(1-\frac{\delta}{\Delta}\right)^{-2}\frac{\delta}{\Delta}<\frac{2\delta}{\Delta}.
\end{equation}
In addition,  combining (\ref{eigenother}) and (\ref{tildevi}) implies
\begin{equation}
\label{viprimevi}
\| v'_i-v_i\|\leq \frac{3N^{2.5}}{\Delta}\cdot \delta.
\end{equation}
 On the other hand, for any $j\in\{1,\ldots,k\}$, we deduce from  (\ref{deltacond}), (\ref{primevi}) and Lemma~\ref{crosspolytope} that there exist an orthornormal basis $\{w_i:\,i\in I_j\}$ of $L_j$ such that
\begin{equation}
\label{primeviwi}
\| w_i-v'_i\|\leq \frac{4N}{\Delta}\cdot \delta\mbox{ \ for $i\in I_j$},
\end{equation}
Combining (\ref{viprimevi}) and (\ref{primeviwi}) yields (\ref{wiviL}).

Next we extend the basis $w_1,\ldots,w_r$  of $L$ in (\ref{wiviL}) into an orthonormal basis $w_1,\ldots,w_N$ of $\R^N$ such that
\begin{equation}
\label{wiviRN}
\|w_i-v_i\|\leq  \frac{28N^4}{\Delta}\cdot \delta
\mbox{ \ for $i=1,\ldots,N$}.
\end{equation}
Having (\ref{wiviL}) at hand, we may assume that $r<N$.  If $1\leq i\leq r$ and $r<j\leq N$, then
$\langle v_j,v_i\rangle =0$ and (\ref{wiviL}) yield that
$|\langle w_i,v_j\rangle|\leq  \frac{7N^3}{\Delta}\cdot \delta$. Therefore if $r<j\leq N$, then
\begin{equation}
\label{vjLLbot}
\begin{array}{rcl}
\left\|v_j|L\right\|&\leq&  \frac{7N^4}{\Delta}\cdot \delta\\[1ex]
\left\|v_j|L^\bot\right\|&\geq&  1-\frac{7N^4}{\Delta}\cdot \delta
\end{array}
\end{equation}
From this point, we follow a similar path as in the case of (\ref{wiviL}). 
For $j=r+1,\ldots,N$, we write $\tilde{v}_j=v_j|L^\bot$ and 
$v'_j=\tilde{v}_j/\|\tilde{v}_j\|\in L^\bot\cap S^{N-1}$,
and hence (\ref{vjLLbot}) yields that 
\begin{equation}
\label{vjprimevj}
\|v_j-v'_j\|\leq  \frac{14N^4}{\Delta}\cdot \delta,
\end{equation}
and if $r+1\leq j<l\leq N$, then (\ref{deltacond}) and (\ref{vjLLbot}) imply that 
\begin{eqnarray}
\nonumber
|\langle v'_i,v'_l\rangle|&\leq& \left(1-\frac{7N^4}{\Delta}\cdot \delta\right)^{-2}
|\langle \tilde{v}_i,\tilde{v}_l\rangle|<2|\langle \tilde{v}_i,\tilde{v}_l\rangle|=2
|\langle (v_j|L),(v_l|L)\rangle|\\
\label{primevjlLbot}
&< &2\cdot\frac{7N^4}{\Delta}\cdot \delta\cdot\frac1{2N}<
\frac{7N^3}{\Delta}\cdot \delta.
\end{eqnarray}
In turn, we conclude from  (\ref{deltacond}), (\ref{primevjlLbot}) and Lemma~\ref{crosspolytope} that there exist an orthornormal basis $w_{r+1},\ldots,w_N$ of $L^\bot$ such that
\begin{equation}
\label{primevjwj}
\| w_j-v'_j\|\leq \frac{14N^4}{\Delta}\cdot \delta\mbox{ \ for $j=r+1,\ldots,N$},
\end{equation}
Combining (\ref{wiviL}), (\ref{vjprimevj}) and (\ref{primevjwj}) yields (\ref{wiviRN}).

We write $\widetilde{D}$ to denote the diagonal matrix whose first $r$ diagonal entries are $\mu_1,\ldots,\mu_r>0$ in this order, and the rest of the entries are $0$ (compare (\ref{muipos})).
For the $N\times N$ orthogonal transformation $F$
defined by $Fw_i=v_i$ for $i=1,\ldots,N$, we have
$$
\widetilde{D}=FMAM^{-1}F^{-1}.
$$
Writing $\widetilde{E}=\sqrt{\widetilde{D}}-\sqrt{D}$, it follows from (\ref{sqrtlambdamu}) that
\begin{equation}
\label{Etilde}
\|\widetilde{E}\|\leq \frac{N^2}{\Delta}\cdot \delta.
\end{equation}
In addition, (\ref{deltacond}) and (\ref{wiviRN}) yield that
\begin{equation}
\label{FId}
\|F^{-1}-{\rm Id}\|\leq  \frac{28N^5}{\Delta}\cdot \delta
\mbox{ \ and \ }
\|F-{\rm Id}\|\leq  \frac{28N^5}{\Delta}\cdot \delta<1.
\end{equation}

For the positive semi-definite $P=M^{-1}F^{-1}\sqrt{\widetilde{D}}FM$ and $Q=M^{-1}\sqrt{D}M$
matrices, we have $A=PP$ and $B=QQ$, and
\begin{eqnarray*}
\|P-Q\|&=&\|M^{-1}F^{-1}\sqrt{\widetilde{D}}FM-M^{-1}\sqrt{D}M\|=
\|M^{-1}(F^{-1}\sqrt{\widetilde{D}}F-\sqrt{D})M\|\\
&\leq& 
\|F^{-1}\sqrt{\widetilde{D}}F-\sqrt{D}\|=\|F^{-1}(\sqrt{D}+\widetilde{E})F-\sqrt{D}\|\\
&\leq &
\|(F^{-1}-{\rm Id})\sqrt{D}(F-{\rm Id})+(F^{-1}-{\rm Id})\sqrt{D}+\sqrt{D}(F-{\rm Id})\| + \|F^{-1}\widetilde{E}F\|\\
&\leq&
\sqrt{\lambda_r}\left(\|F^{-1}-{\rm Id}\|\cdot \|F-{\rm Id}\|+ \|F^{-1}-{\rm Id}\| + \|F-{\rm Id}\| \right)+\|\widetilde{E}\|
\end{eqnarray*}
Combining (\ref{Etilde}) and (\ref{FId}) implies
$$
\|P-Q\|\leq \frac{85N^5\max\{\sqrt{\|B\|},1\}}{\Delta}\cdot \delta,
$$
thus we may choose $K=\frac{85N^5\max\{\sqrt{\|B\|},1\}}{\Delta}$.
\end{proof}

Theorem \ref{thm:general} immediately follows from Theorems \ref{thm:weak-stability} and \ref{thm:code-stability}.

\begin{proof}[Proof of Theorem \ref{thm:general}]
We use notation from Theorem \ref{thm:weak-stability}. The Gram matrix $A$ of any $(d,N,s+\varepsilon)$-code is $K\varepsilon^{1/m}$-close to the Gram matrix $B$ of a Delsarte-tight $(d,N,s)$-code for sufficiently small $\varepsilon$. Then, by Theorem \ref{thm:code-stability}, $B$ and $A$ can be represented as $QQ$ and $PP$, respectively, where $Q$ and $P$ are positive semi-definite symmetric matrices which are $\frac{85N^5\max\{\sqrt{\|B\|},1\}}{\Delta} K \varepsilon^{1/m}$-close, where $\Delta$ is the minimum of the minimal gap between two consecutive eigenvalues of $B$ and the minimal positive eigenvalue of $B$. The matrices $P$ and $Q$ define by their column vectors two $d$-dimensional spherical codes whose Gram matrices are $A$ and $B$, respectively. The corresponding columns $P^i$ and $Q^i$, $1\leq i\leq N$, are unit $d$-dimensional vectors whose components differ by no more than $\frac{85N^5\max\{\sqrt{\|B\|},1\}}{\Delta} K \varepsilon^{1/m}$ so $\|P^i-Q^i\|\leq \sqrt{d}\cdot \frac{85N^5\max\{\sqrt{\|B\|},1\}}{\Delta} K \varepsilon^{1/m}$. Using $\angle(P^i,Q^i)\leq \frac {\pi} 2 \|P^i-Q^i\|$, we get that the sphecrical distance between $P^i$ and $Q^i$ is not greater than $C\varepsilon^{1/m}$, where $C=\frac {\pi} 2 \sqrt{d}\cdot \frac{85N^5\max\{\sqrt{\|B\|},1\}}{\Delta} K$.
\end{proof}

\section{Stability of the $(8,240,1/2)$-code}\label{sect:8}

The set of minimal vectors of the $E_8$ lattice forms the Delsarte-tight $(8,240,1/2)$-code. The polynomial $f$ for this code is

$$f(t)=\frac {320} {3} (t + 1)(t + 1/2)^2 t^2 (t-1/2) =$$

$$= Q_0 + \frac {16} 7 Q_1 + \frac {200} {63} Q_2 + \frac {832} {231} Q_3 + \frac {1216}{429} Q_4 + \frac {5120} {3003} Q_5 + \frac {2560}{4641} Q_6,$$

where $f_0=1$ and $f(1)=240$.

It follows from the proof of Theorem \ref{thm:weak-stability} that for any $\varepsilon<5\cdot 10^{-6}$ any $(8,N',1/2+\varepsilon)$-code must have no more than $240$ points. For the weak stability, slightly generalizing the outcome of Theorem \ref{thm:weak-stability}, we get that for any $\varepsilon < \frac 1 {(6\cdot 10^7\cdot 240!\cdot 2^{240})^2}$ the combinatorial structure of a $(8,240,1/2+\varepsilon)$-code is the same as for the $(8,240,1/2)$-code and all dot products in a $(8,240,1/2+\varepsilon)$-code differ by not more than $3\cdot 10^7\sqrt{\varepsilon}$ from $0$ or $-1/2$ or by not more than $3\cdot 10^7\varepsilon$ from $-1$ or $1/2$. From now on we consider only $(8,240,1/2+\varepsilon)$-codes as described above. In what follows, we will show that the linear programming approach and the combinatorial structure of the code force all dot products to be within $O(\varepsilon)$ of the dot products of the $(8,240,1/2)$-code.

For any two points $x, y$ in the $(8,240,1/2)$-code such that $\langle x,y\rangle=-1/2$, there is a point $z$ such that $\langle x,z\rangle=\langle y,z\rangle=1/2$. This is true because $x+y$ must belong to $E_8$ as well, has length 1 and forms the angles of $\pi/3$ with both $x$ and $y$. From here we conclude that for any two points $x,y$ of a $(8,240,1/2+\varepsilon)$-code such that $\langle x,y\rangle$ is close to $-1/2$ there exists a point $z$ with $\langle x,z\rangle$ and $\langle y,z\rangle$ not smaller than $1/2 - 3\cdot 10^7\varepsilon$. From the triangle inequality for spherical distances, we get that $\langle x,y\rangle \geq -1/2 -  9\cdot 10^7\varepsilon$.

Similarly to the proof of Theorem \ref{linprogbound}, for any $(d,N,s)$-code $X$, any $f=f_0Q_0+f_1Q_1+\ldots+f_kQ_k$ with non-negative coefficients, and any $i\in [1,k]$,

\begin{equation}
\label{linprogbound-i}
Nf(1)+\sum_{x,y\in X\atop x\neq y}f(\langle x,y\rangle)\geq N^2 f_0 + \sum_{x,y\in X} f_i Q_i(\langle x,y\rangle)\geq N^2 f_0.
\end{equation}

For $N$ and $f$ satisfying the conditions of a Delsarte-tight code, $Nf(1)=N^2 f_0$ so we get

\begin{equation}\label{qi-bound}
\sum_{x,y\in X\atop x\neq y}f(\langle x,y\rangle)\geq \sum_{x,y\in X} f_i Q_i(\langle x,y\rangle)\geq 0.
\end{equation}

As we know, $f(\langle x,y\rangle)$ for $x\neq y$ in our code is either non-positive or not greater than $\max\limits_{t\in[1/2,1]} \left\{\frac {f(t)} {t-1/2}\right\} \varepsilon = 480\varepsilon$. Therefore, for each $i\in[1,6]$,

\begin{equation}\label{qi-bound-8}
0\leq \sum_{x,y\in X} Q_i(\langle x,y\rangle) \leq \frac 1 {f_i} (240^2-240)\cdot 480\varepsilon.
\end{equation}

We note that for the $(8,240,1/2)$-code all six sums are 0.

We use the inequality $0\leq \sum_{x,y\in X} Q_2(\langle x,y\rangle)$. For $d=8$, $Q_2(t)=\frac 8 7 t^2 - \frac 1 7$. If $\langle x,y\rangle$ is close to $1/2$, $Q_2(\langle x,y\rangle)\leq Q_2(1/2) + 7\cdot 10^7\varepsilon$. If $\langle x,y\rangle$ is close to $0$, $Q_2(\langle x,y\rangle)\leq Q_2(0) + 2\cdot 10^{15}\varepsilon$. If $\langle x,y\rangle$ is close to $-1$, $Q_2(\langle x,y\rangle)\leq Q_2(-1)$. If $\langle x,y\rangle$ is close to $-1/2$ and less than $-1/2$, $Q_2(\langle x,y\rangle)\leq Q_2(-1/2) + 3\cdot 10^8\varepsilon$. Overall, for all considered cases if $\langle x,y\rangle$ is close to $\alpha$, $Q_2(\langle x,y\rangle)- Q_2(\alpha) \leq 2\cdot 10^{15}\varepsilon$. Since the sum for $Q_2$ is 0 on the $(8,240,1/2)$-code, the total sum of $Q_2(\langle x,y\rangle)- Q_2(\alpha)$ is non-negative too. The sum of non-negative terms from here is not greater than $(240^2-240)\cdot 2 \cdot 10^{15}\varepsilon< 2\cdot 10^{20}\varepsilon$. The only unobserved case so far is the one when $\langle x,y\rangle$ is close to $-1/2$ and not smaller than $-1/2$. In this case, $Q_2(\langle x,y\rangle)- Q_2(-1/2)$ is non-positive and cannot be larger by its absolute value than the sum of all positive elements. Therefore,

$$Q_2(\langle x,y\rangle)- Q_2(-1/2) \geq - 2\cdot 10^{20}\varepsilon,$$

$$\langle x,y\rangle \leq -1/2 + 2\cdot 10^{20}\varepsilon.$$

In the $(8,240,1/2)$-code for any pair of points $x,y$ such that $\langle x,y\rangle=0$, there exist 12 more points $z$ such that $\langle x,z\rangle = \langle y,z\rangle = 1/2$ (due to the association scheme structure this number is 12 for any such pair of $x$ and $y$). Note that all these points including $x$ and $y$ belong to the $6$-dimensional sphere of radius $\frac 1 {\sqrt{2}}$ with the center at $\frac {x+y} 2$. 14 points on the same $6$-dimensional sphere of radius $\frac 1 {\sqrt{2}}$ with the minimal distance equal to $1=\sqrt{2}\cdot \frac 1 {\sqrt{2}}$ must be located in the vertices of the $7$-dimensional cross-polytope. All pairs of points with $\langle x,y\rangle=0$ are then partitioned into 7-tuples from the same $7$-dimensional cross-polytope.

Consider two pairs of points $x,y$ and $z,t$ from the $(8,240,1/2+\varepsilon)$-code $X$ such that $\langle x,y\rangle$ and $\langle z,t\rangle$ are close to 0 and the other four dot products between pairs of these points are close to $1/2$. We denote $\langle x,y\rangle$ by $\alpha$ and $\langle z,t\rangle$ by $\beta$. Assume also that $\langle x,z\rangle = 1/2+\delta_1$, $\langle x,t\rangle=1/2+\delta_2$, $\langle y,z\rangle=1/2+\delta_3$, $\langle y,t\rangle = 1/2+\delta_4$. All $|\delta_i|$, $1\leq i\leq 4$, are not greater than $3\cdot 10^7\varepsilon$. The Gram matrix for these four points should be positive semi-definite so its determinant is non-negative:

\begin{equation}\label{0-det}
0\leq
\begin{vmatrix}
   1 & \alpha & \frac 1 2 +\delta_1 & \frac 1 2 + \delta_2 \\
    \alpha & 1 & \frac 1 2 +\delta_3 & \frac 1 2 +\delta_4 \\
    \frac 1 2 +\delta_1 & \frac 1 2 +\delta_3 & 1 & \beta \\
    \frac 1 2 +\delta_2 & \frac 1 2 +\delta_4 & \beta & 1
\end{vmatrix}
\leq
\begin{vmatrix}
   1 & \alpha & \frac 1 2 & \frac 1 2 \\
    \alpha & 1 & \frac 1 2 & \frac 1 2 \\
    \frac 1 2 & \frac 1 2 & 1 & \beta \\
    \frac 1 2 & \frac 1 2 & \beta & 1
\end{vmatrix}
+2\cdot 10^9 \varepsilon=
\end{equation}

$$=(1-\alpha)(1-\beta)(\alpha\beta+\alpha+\beta) + 2\cdot 10^9 \varepsilon.$$

Hence we get that $\alpha\beta+\alpha+\beta\geq -3\cdot 10^9 \varepsilon$. Since $\alpha\beta\leq 9\cdot 10^{14} \varepsilon$, $\alpha+\beta \geq - 10^{15} \varepsilon$. For each 7-dimensional cross-polytope described above we can average these inequalities over all pairs and get that the sum of the seven dot products in it is at least $-35\cdot 10^{14}\varepsilon$. On the other hand, if we set aside one pair of opposite vertices in a cross-polytope and average over six pairs , we will get that their sum is at least $-30\cdot 10^{14}\varepsilon$. Overall, if we set aside one fixed pair of points $x_0,y_0$ such that $\langle x_0,y_0\rangle$ is close to 0, then all other pairs of such points on average has a dot product at least $-5\cdot 10^{14}\varepsilon$

Using inequality (\ref{qi-bound-8}) for the Gram matrix (matrix for $Q_1$) we get that

$$0\leq \sum_{x,y\in X} \langle x,y\rangle \leq 2\cdot 10^7 \varepsilon.$$

For all pairs $x,y$, where $\langle x,y\rangle$ is close to $-1$, $-1/2$, $1/2$, $\langle x,y\rangle$ differs from all these numbers by no more than $2\cdot 10^{20}\varepsilon$. If we fix one pair of points $x_0,y_0$ such that $\langle x_0,y_0\rangle$ is close to 0, then, using all lower bounds above, we get

$$\langle x_0,y_0\rangle - (240^2-240-1)\cdot 2\cdot 10^{20}\varepsilon \leq 2\cdot 10^7 \varepsilon,$$

$$\langle x_0,y_0\rangle \leq 2\cdot 10^{25}  \varepsilon.$$

Hence for all $x,y$ such that $\langle x,y\rangle$ is close to 0, $\langle x,y\rangle \leq 2\cdot 10^{25}  \varepsilon$. We can use this to bound $\langle x,y\rangle$ from below as well, using the inequality $\alpha+\beta \geq - 10^{15} \varepsilon$: $\langle x,y\rangle \geq - 10^{15} \varepsilon - 2\cdot 10^{25}  \varepsilon\geq - 3\cdot 10^{25}  \varepsilon$.

Combining all the bounds we obtained, the $(8,240,1/2+\varepsilon)$-code is $3\cdot 10^{25}\varepsilon$-close to the $(8,240,1/2)$-code. Coupling this with Theorem \ref{thm:code-stability} we get Theorem \ref{thm:8}. The proof follows the proof of Theorem \ref{thm:general} almost word-for-word. The constant $C_8$ we obtain satisfies

$$C_8\geq \frac {\pi} 2 \sqrt{8}\cdot \frac{85\cdot 240^5\max\{\sqrt{\|B\|},1\}}{\Delta}\cdot 3\cdot 10^{25},$$

where $B$ is the Gram matrix of the $(8,240,1/2)$-code and $\Delta$ is the minimum of the minimal gap between two consecutive eigenvalues of $B$ and the minimal positive eigenvalue of $B$. $\|B\|$ is not greater than 240 and $\Delta$ may be estimated by using the root separation bounds from \cite{rum79}. Overall, we get that it is sufficient to take $C_8=10^{206}$.

\section{Stability of the $(24,196560,1/2)$-code}\label{sect:24}

The set of minimal vectors of the Leech lattice forms the Delsarte-tight $(24,196560,1/2)$-code with the polynomial

$$f(t)= \frac {1490944}{15} (t+1) (t+1/2)^2(t+1/4)^2t^2(t-1/4)^2(t-1/2)=$$

$$= Q_0 +\frac {48}{23} Q_1 + \frac {1144} {425} Q_2 + \frac {12992}{3825} Q_3 + \frac {73888}{22185} Q_4 + \frac {2169856}{687735} Q_5 +$$

$$+ \frac {59062016}{25365285} Q_6 + \frac {4472832}{2753575} Q_7 + \frac {23855104}{28956015} Q_8 + \frac {7340032}{20376455} Q_9 + \frac {7340032}{80848515} Q_{10},$$

so that $f_0=1$ and $f(1)=196560$.

The structure of the proof will be generally similar to the one for the $(8,240,1/2)$-code. Firstly, we use the proof of Theorem \ref{thm:weak-stability} and see that for $\varepsilon<10^{-11}$ any $(24,N',1/2+\varepsilon)$-code with $N'\geq 196560$ must have precisely $196560$ points. The proof of Theorem \ref{thm:weak-stability} also implies that, when $\varepsilon< \frac 1 {(4\cdot 10^{16} \cdot 196560!\cdot 4^{196560})^2}$, the structure of this code is combinatorially the same as the one of the $(24,196560,1/2)$-code and all inner products close to $1/2$ and $-1$ must be within $2\cdot 10^{16}\varepsilon$ from $1/2$ and $-1$, respectively, all inner products close to $-1/4$, $-1/2$, $0$, $1/4$ must be within $2\cdot 10^{16}\sqrt{\varepsilon}$ from $-1/4$, $-1/2$, $0$, $1/4$, respectively. From now on we consider only codes with these constraints on dot products. For the main part of the proof, we will use the linear programming approach and the combinatorics of the $(24,196560,1/2)$-code to show that all inner products of such $(24,196560,1/2+\varepsilon)$-codes are, in fact, within $O(\varepsilon)$ of their counterparts among minimal vectors of the Leech lattice.

Inequality (\ref{qi-bound}) is true for any case of a Delsarte-tight code so we can find an analogue of inequality (\ref{qi-bound-8}) in the 24-dimensional case. In a $(24,196560,1/2)$-code $X$, $f(\langle x,y\rangle)$ for $x,y\in X$, $x\neq y$, is either non-positive or not greater than $\max\limits_{t\in[1/2,1]} \left\{\frac {f(t)} {t-1/2}\right\} \varepsilon = 393120 \varepsilon$. Therefore, for each $i\in[1,10]$,

\begin{equation}\label{qi-bound-24}
0\leq \sum_{x,y\in X} Q_i(\langle x,y\rangle) \leq \frac 1 {f_i} (196560^2-196560)\cdot 393120\varepsilon < 2\cdot 10^{17} \varepsilon.
\end{equation}

Similarly to the 8-dimensional case, we note that these sums are identically 0 if $X$ is the unique $(24,196560,1/2)$-code.

We will split all pairs of points from $X^2$ into groups $A_\alpha$, $\alpha=-1$, $-1/2$,$ -1/4$, $0$, $1/4$, $1/2$, $1$, such that $(x,y)\in A_\alpha$ if $\langle x,y\rangle$ is close to $\alpha$ (the corresponding pair of points in the $(24,196560,1/2)$-code has inner product~$\alpha$). By $S_\alpha$ we denote $\sum_{(x,y)\in A_\alpha} (\langle x,y\rangle -\alpha)$. By using inequalities (\ref{qi-bound-24}), we will show that all $S_\alpha$ are $O(\varepsilon)$.

This immediately holds for $S_{1}$ (since it is identically 0), $S_{1/2}$, $S_{-1}$:

$$-(196560^2-196560)\cdot 2 \cdot 10^{16} \varepsilon \leq S_{1/2}\leq (196560^2-196560)\varepsilon \text{, hence } |S_{1/2}|\leq 8\cdot 10^{26} \varepsilon;$$

$$0\leq S_{-1} \leq (196560^2-196560)\cdot 2 \cdot 10^{16} \varepsilon  \text{, hence } |S_{-1}|\leq 8\cdot 10^{26} \varepsilon.$$

We use inequalities (\ref{qi-bound-24}) for $i=1, 2, 3, 4$ approximating $Q_i(\langle x,y\rangle)$ by $Q_i(\alpha)+(\langle x,y\rangle -\alpha) Q'_i(\alpha)$. We also change the left and right sides of the inequalities to $\pm 10^{45}\varepsilon$ in order to cover deficiencies caused by omitting higher order terms: $(\langle x,y\rangle -\alpha)^2\leq 4\cdot 10^{32}\varepsilon$, all terms with degree at least 3 are much smaller because they are $O(\varepsilon^{3/2})$ and $\varepsilon$ is extremely small. 

$$-10^{45}\varepsilon\leq S_{-1}+S_{-1/2}+S_{-1/4}+S_0+S_{1/4}+S_{1/2}\leq 10^{45}\varepsilon;$$

$$-10^{45}\varepsilon\leq - \frac {48} {23} S_{-1} - \frac {24} {23} S_{-1/2} -\frac {12} {23} S_{-1/4} +\frac {12} {23} S_{1/4} + \frac {24} {23} S_{1/2} \leq 10^{45}\varepsilon;$$

$$-10^{45}\varepsilon\leq \frac {75} {23} S_{-1} + \frac {33} {46} S_{-1/2} + \frac {15} {184} S_{-1/4} -3 S_0 + \frac {15} {184} S_{1/4} + \frac {33} {46} S_{1/2} \leq 10^{45}\varepsilon;$$

$$-10^{45}\varepsilon\leq - \frac {104} {23} S_{-1} -\frac {208} {575} S_{-1/2} - \frac {13} {230} S_{-1/4} + \frac {13} {230} S_{1/4} + \frac {208} {575} S_{1/2} \leq 10^{45}\varepsilon.$$

This system of inequalities must imply that all $S_\alpha$ are $O(\varepsilon)$ because the coefficients for $S_{-1/2}$, $S_{-1/4}$, $S_0$, $S_{1/4}$ form a non-singular matrix. More precisely, from the second and forth inequalities we can immediately get that $|S_{-1/2}|\leq 10^{46} \varepsilon$ and $|S_{1/4}-S_{-1/4}|\leq 2\cdot 10^{46} \varepsilon$. From the first and the third inequality we then get that $|S_{1/4}+S_{-1/4}|\leq 2\cdot 10^{46} \varepsilon$ too so both $|S_{1/4}|$ and $|S_{-1/4}|$ are not greater than $2\cdot 10^{46} \varepsilon$. Using the bound for $|S_{1/4}+S_{-1/4}|$ we also find that $|S_0|\leq 4\cdot 10^{46} \varepsilon$.

Now we will use the bound on $S_0$ to show that, for each pair of points $x,y\in X$ such that $\langle x,y \rangle$ is close to 0, $\langle x,y \rangle$ is $O(\varepsilon)$. The proof is similar to the one for the 8-dimensional kissing configuration.

In the $(24,196560,1/2)$-code for any pair of points $x,y$ such that $\langle x,y\rangle=0$, there exist 44 more points $z$ such that $\langle x,z\rangle = \langle y,z\rangle = 1/2$ (due to the association scheme structure this number is 44 for any such pair of $x$ and $y$). All these points including $x$ and $y$ belong to the $22$-dimensional sphere of radius $\frac 1 {\sqrt{2}}$ with the center at $\frac {x+y} 2$. 46 points on the same $22$-dimensional sphere of radius $\frac 1 {\sqrt{2}}$ with the minimal distance equal to $1$ must be the vertices of the $23$-dimensional cross-polytope. All pairs of points with $\langle x,y\rangle=0$ are then partitioned into 23-tuples from the same $23$-dimensional cross-polytope.

For two pairs of points $x,y$ and $z,t$ from $X$ such that $\langle x,y\rangle$ and $\langle z,t\rangle$ are close to 0 and the other four dot products between these points are close to $1/2$, denote $\langle x,y\rangle$ by $\alpha$ and $\langle z,t\rangle$ by $\beta$. Assume also that $\langle x,z\rangle = 1/2+\delta_1$, $\langle x,t\rangle=1/2+\delta_2$, $\langle y,z\rangle=1/2+\delta_3$, $\langle y,t\rangle = 1/2+\delta_4$. All $\delta_i$, $1\leq i\leq 4$, are not greater than $2\cdot 10^{16}\varepsilon$. The Gram matrix for these four points should be positive semi-definite so its determinant is non-negative:

\begin{equation}\label{0-det}
0\leq
\begin{vmatrix}
   1 & \alpha & \frac 1 2 +\delta_1 & \frac 1 2 + \delta_2 \\
    \alpha & 1 & \frac 1 2 +\delta_3 & \frac 1 2 +\delta_4 \\
    \frac 1 2 +\delta_1 & \frac 1 2 +\delta_3 & 1 & \beta \\
    \frac 1 2 +\delta_2 & \frac 1 2 +\delta_4 & \beta & 1
\end{vmatrix}
\leq
\begin{vmatrix}
   1 & \alpha & \frac 1 2 & \frac 1 2 \\
    \alpha & 1 & \frac 1 2 & \frac 1 2 \\
    \frac 1 2 & \frac 1 2 & 1 & \beta \\
    \frac 1 2 & \frac 1 2 & \beta & 1
\end{vmatrix}
+10^{18} \varepsilon=
\end{equation}

$$=(1-\alpha)(1-\beta)(\alpha\beta+\alpha+\beta) + 10^{18} \varepsilon.$$

Hence we get that $\alpha\beta+\alpha+\beta\geq -10^{18} \varepsilon$. Since $\alpha\beta\leq 4\cdot 10^{32} \varepsilon$, $\alpha+\beta \geq - 6\cdot 10^{32} \varepsilon$. Averaging this inequality over all pairs from a 23-dimensional cross-polytope described above we get the the sum of dot products for pairs of opposite vertices in such a cross-polytope is at least $-23\cdot 3 \cdot 10^{32} \varepsilon$. On the other hand, averaging over all pairs from a cross-polytope, except for a fixed pair $(x_0,y_0)$, the sum of dot products is at least $-22\cdot 3\cdot 10^{32} \varepsilon$. From these inequalities,

$$S_0\geq \langle x_0, y_0 \rangle - (196560^2-196560-1)\cdot 3\cdot 10^{32} \varepsilon \geq \langle x_0, y_0 \rangle - 2\cdot 10^{43}\varepsilon.$$

Combining this with the bound on $S_0$, we find that $| \langle x_0, y_0 \rangle| \leq 5\cdot 10^{46} \varepsilon$.

For the next step we will show that if $\langle x_0, y_0 \rangle$ is close to -1 then $\langle x_0, y_0 \rangle$ differs from -1 by $O(\varepsilon^2)$. In order to do this we use the following lemma.

\begin{lemma}\label{lem:almost_perp}
In a $d$-dimensional spherical code $\{x, z_1, z_2,\ldots, z_{d-1}\}$, $|\langle x, z_i \rangle| < \delta$ for all $i$ from 1 to $d-1$ and for a fixed $\delta>0$. If $z_1, z_2,\ldots, z_{d-1}$ are linearly independent, then, for one of the two unit vectors orthogonal to the span of $z_1, z_2,\ldots, z_{d-1}$, its spherical distance to $x$ is not greater than $\frac \pi 2 \sqrt{ \frac {d-1} {\lambda_1(G)}} \delta$, where $\lambda_1(G)$ is the minimal eigenvalue of the Gram matrix $G$ of $z_1, z_2,\ldots, z_{d-1}$.
\end{lemma}

\begin{proof}
We denote one of the unit vectors orthogonal to the span of $z_1, z_2,\ldots, z_{d-1}$ by $z_\perp$ if $\langle x, z_\perp \rangle \geq 0$. We represent $x$ as $\alpha_\perp z_\perp + x_z$, where $x_z$ belongs to the span of $z_1, z_2,\ldots, z_{d-1}$. Then $|\langle x_z, z_i \rangle| < \delta$ for all $i$ from 1 to $d-1$. We denote $\langle x_z, z_i \rangle$ by $\Delta_i$ for all $i$ and the vector $(\Delta_1,\ldots,\Delta_{d-1})$ by $\Delta$. If $\alpha=(\alpha_1,\ldots,\alpha_{d-1})$ is the vector of coordinates of $x_z$ in the basis $z_1, z_2,\ldots, z_{d-1}$, then $\Delta=G\alpha$ and $\alpha=G^{-1} \Delta$. On the other hand,

$$||x_z||^2=\alpha^T G \alpha \leq ||\alpha||\cdot ||\Delta||=||G^{-1} \Delta||\cdot ||\Delta||\leq \frac 1 {\lambda_1(G)} ||\Delta||^2 \leq \frac {d-1} {\lambda_1(G)} \delta^2.$$

Therefore, $\angle(z_\perp,x) = \arcsin{||x_z||} \leq \frac \pi 2 ||x_z||\leq \frac \pi 2 \sqrt{ \frac {d-1} {\lambda_1(G)}} \delta$.
\end{proof}

We consider an arbitrary pair $x_0$, $y_0$ of points from $X$ such that $\langle x_0, y_0 \rangle$ is close to -1. Among 93150 points $x$ from $X$ such that $\langle x_0, x \rangle$ is close to 0, we choose arbitrarily 23 points so that their counterparts in the $(24,196560,1/2)$-code form a basis. Since all dot products in this basis are $\pm 1/2, \pm 1/4, 0$, the Gram matrix of this basis has determinant not smaller than $\frac 1 {4^{23}}$. We can deduce that the Gram matrix of the corresponding basis in $X$ has determinant not smaller than $\frac 1 {4^{24}}$. Each eigenvalue of the Gram matrix is not greater than $13$ (for instance, by the Gershgorin circle theorem). Hence the smallest one is at least $\frac 1 {4^{24} 13^{23}}$. Lemma \ref{lem:almost_perp} implies that the angle between $x_0$ and $y_0$ must be at least

$$\pi - \pi \sqrt{\frac {23} {\lambda_1(G)}}\cdot 5\cdot 10^{46} \varepsilon \geq \pi - \pi \sqrt{23\cdot 4^{24}\cdot 13^{23}}\cdot 5\cdot 10^{46} \varepsilon \geq \pi - 10^{68}\varepsilon.$$

For any points $x,y\in X$ such that $\langle x,y \rangle$ is close to -1/2, we consider a point $z\in X$ such that $\langle x,z \rangle$ is close to -1 and $\langle y,z \rangle$ is close to 1/2. Then, by the inequality on $\angle (x,z)$ proven above, $\langle x,y \rangle$ differs from $-\langle y,z \rangle$ by no more than $10^{68}\varepsilon$. Given that $\langle y,z \rangle$ is within $2\cdot 10^{16}\varepsilon$ of 1/2 we can conclude that  $\langle x,y \rangle$ is within $2\cdot 10^{68}\varepsilon$ of -1/2.

For the next step, consider a pair of points $x'$ and $y'$ of the $(24,196560,1/2)$-code such that $\langle x',y' \rangle=1/4$. There are exactly 275 points $u'$ of the $(24,196560,1/2)$-code such that $\langle x',u' \rangle=\langle y',u' \rangle=1/2$ (this number is the same for all such pairs $x', y'$ due to the association scheme structure). It is not hard to see that, after the appropriate dilation, these 275 points form a $(22, 275, 1/6)$-code. This is a Delsarte-tight code in dimension 22 with exactly two inner products, 1/6 and -1/4. It must possess the structure of a strongly regular graph (see \cite{del77}). Due to this structure, there is a unique $(22, 275, 1/6)$-code \cite{goe75}. In what follows we will analyze the counterpart of this code in the $(24,196560,1/2+\varepsilon)$-code $X$.

Consider two points $x, y\in X$ such that $\langle x,y \rangle$ is close to 1/4. We denote $\langle x,y \rangle$ by $1/4+\delta$, where $|\delta|$ is known to be not greater than $2\cdot 10^{16}\sqrt{\varepsilon}$. With the slight abuse of notation, by $\pm t$, for any real $t$, we will mean an unknown real number between $-t$ and $t$. For instance, for each point $u\in X$ such that both $\langle x,u \rangle$ and $\langle y,u \rangle$ are close to 1/2, we can write $\langle x,u \rangle =1/2\pm 2\cdot 10^{16}\varepsilon$ and $\langle y,u \rangle = 1/2\pm 2\cdot 10^{16}\varepsilon$. Any point $u$ of this kind may be uniquely represented as $\alpha x + \beta y + \gamma z$, where $z$ is a unit vector orthogonal both to $x$ and $y$ and $\alpha, \beta, \gamma$ are real with $\gamma>0$. Straightforward calculations show that both $\alpha$ and $\beta$ must be $2/5 \pm 3\cdot 10^{16}\varepsilon$. Since $||u||=1$, we get

$$\alpha^2+2\alpha\beta \left(\frac 1 4 +\delta\right) +\beta^2 +\gamma^2=1,$$

$$\frac 2 5 + \frac 8 {25} \delta +\gamma^2 \pm 7\cdot 10^{16} \varepsilon = 1,$$

$$\gamma = \sqrt{\frac 3 5} - \frac {4\sqrt{15}} {75}\delta \pm 2\cdot 10^{31} \varepsilon.$$

Now we consider two points $u_1, u_2\in X$ satisfying the conditions above. Hence for their representations $u_1=\alpha_1 x +\beta_1 y + \gamma_1 z_1$ and $u_2=\alpha_2 x +\beta_2 y + \gamma_2 z_2$, $\alpha_1=2/5 \pm 3\cdot 10^{16}\varepsilon$, $\beta_1=2/5 \pm 3\cdot 10^{16}\varepsilon$, $\gamma_1=\sqrt{\frac 3 5} - \frac {4\sqrt{15}} {75}\delta \pm 2\cdot 10^{31} \varepsilon$, $\alpha_2=2/5 \pm 3\cdot 10^{16}\varepsilon$, $\beta_2=2/5 \pm 3\cdot 10^{16}\varepsilon$, $\gamma_2=\sqrt{\frac 3 5} - \frac {4\sqrt{15}} {75}\delta \pm 2\cdot 10^{31} \varepsilon$. Then we can estimate $\langle u_1, u_2 \rangle$:

$$\langle u_1, u_2 \rangle = \alpha_1\alpha_2 + (\alpha_1\beta_2+\beta_1\alpha_2)\left(\frac 1 4 +\delta \right) +\beta_1\beta_2 + \gamma_1\gamma_2 \langle z_1, z_2\rangle =$$

$$= \frac 2 5 + \frac 8 {25} \delta + \left( \frac 3 5 - \frac 8 {25} \delta \right) \langle z_1, z_2\rangle \pm 6\cdot 10^{32}\varepsilon.$$

We know that $\langle u_1, u_2 \rangle$ is not greater than $1/2+\varepsilon$. Therefore,

$$\left( \frac 3 5 - \frac 8 {25} \delta \right) \langle z_1, z_2\rangle \leq \frac 1 {10} - \frac 8 {25} \delta + 7\cdot 10^{32}\varepsilon.$$

This inequality must hold for all pairs from 275 points $z$ from the unit sphere in $\mathbb{R}^{22}$. From the tightness of the $(22, 275, 1/6)$-code, this may happen only if for some of these pairs $\langle z_1, z_2\rangle \geq \frac 1 6$. We conclude that

$$\left( \frac 3 5 - \frac 8 {25} \delta \right) \frac 1 6 \leq \frac 1 {10} - \frac 8 {25} \delta + 7\cdot 10^{32}\varepsilon.$$

Subsequently, $\delta \leq 3\cdot 10^{33}\varepsilon$. Hence we proved that for any two points $x,y\in X$ such that $\langle x,y \rangle$ is close to 1/4, $\langle x,y \rangle\leq 1/4 + 3\cdot 10^{33}\varepsilon$. There are no more than $196560^2-196560$ pairs like this so, using that $|S_{1/4}|\leq 2\cdot 10^{46}\varepsilon$, we also get

$$\langle x,y \rangle\geq 1/4 +S_{1/4} - (196560^2-196560)\cdot 3\cdot 10^{33}\varepsilon\geq$$

$$\geq 1/4 - 2\cdot 10^{46}\varepsilon - (196560^2-196560)\cdot 3\cdot 10^{33}\varepsilon \geq 1/4 - 3\cdot 10^{46} \varepsilon.$$

This means that if $\langle x,y \rangle$ is close to 1/4, $\langle x,y \rangle$ differs from 1/4 by no more than $3\cdot 10^{46} \varepsilon$.

For any points $x,y\in X$ such that $\langle x,y \rangle$ is close to -1/4, we consider a point $z\in X$ such that $\langle x,z \rangle$ is close to -1 and $\langle y,z \rangle$ is close to 1/4. Then, by the inequality on $\angle (x,z)$ we proved, $\langle x,y \rangle$ differs from $-\langle y,z \rangle$ by no more than $10^{68}\varepsilon$. Given that $\langle y,z \rangle$ is within $3\cdot 10^{46} \varepsilon$ of 1/4 we can conclude that  $\langle x,y \rangle$ is within $2\cdot 10^{68}\varepsilon$ of -1/4.

Combining all the results from this section, we have shown that any $(24, 196560, 1/2+\varepsilon)$-code $X$ and the unique $(24, 196560, 1/2)$-code are $2\cdot 10^{68}\varepsilon$-close. Together with Theorem \ref{thm:code-stability} this gives the proof of Theorem \ref{thm:24}. The proof is very similar to the proofs of Theorem \ref{thm:general} and Theorem \ref{thm:8}. The constant $C_{24}$ should satisfy

$$C_{24}\geq \frac {\pi} 2 \sqrt{24}\cdot \frac{85\cdot 196560^5\max\{\sqrt{\|B\|},1\}}{\Delta}\cdot 2\cdot 10^{68},$$

where $B$ is the Gram matrix of the $(24, 196560, 1/2)$-code and $\Delta$ is the minimum of the minimal gap between two consecutive eigenvalues of $B$ and the minimal positive eigenvalue of $B$. $\|B\|$ is not greater than 196560 and $\Delta$ may be estimated by using the root separation bounds from \cite{rum79}. Overall, it is sufficient to take $C_{24}=10^{3120}$.

\end{document}